\documentclass{jms}

\def\P{\hbox{\sf P}}
\def\E{\hbox{\sf E}}
\def\1{\hbox{\sf 1}}

\paperTitle{Modified Erlang loss system for cognitive wireless networks}

\authorsShort{E.\,V.~Morozov, S.\,S.~Rogozin, H.\,Q.~Nguyen, T.~Phung-Duc}
\authorsFull{E.\,V.~Morozov\first\also\second\also\third, S.\,S.~Rogozin\first\also\second, H.\,Q.~Nguyen\fourth and T.~Phung-Duc\fourth}
\addAuthorInfo{Institute of Applied Mathematical Research, Karelian Research Centre of RAS, Petrozavodsk, Russia;\\ e-mail: \url{emorozov@karelia.ru} (E.\,M.), \url{ppexa@mail.ru} (S.\,R.);}
\addAuthorInfo{Petrozavodsk State University, Petrozavodsk, Russia;}
\addAuthorInfo{Moscow Center for Fundamental and Applied Mathematics, Moscow State University, Moscow 119991, Russia;}
\addAuthorInfo{University of Tsukuba, 1-1-1 Tennodai, Tsukuba, Ibaraki 305-8573, Japan;\\ e-mail: \url{s2030113@s.tsukuba.ac.jp} (H.\,Q.\,N.), \url{tuan@sk.tsukuba.ac.jp} (T.\,P.)}

\paperAbstract{This paper considers a modified Erlang loss system for cognitive wireless networks and related applications. A primary user has preemptive priority over secondary users and the primary customer is lost if upon arrival all the channels are used by other primary users. Secondary users cognitively use idle channels and they can wait at an infinite buffer in cases idle channels are not available upon arrival or they are interrupted by primary users. We obtain explicit stability condition for the cases where arrival processes of primary users and secondary users follow Poisson processes and their service times follow two distinct arbitrary distributions. The stability condition  is insensitive to the service time distributions and implies the maximal throughout of secondary users. For a special case of exponential service time distributions, we analyze in depth to show the effect of parameters on the delay performance and the mean number of interruptions of secondary users. Our simulations for distributions rather than exponential reveal that the mean number of terminations for secondary users is less sensitive to the service time distribution of primary users.}

\begin{document}
\maketitle

\section{Introduction}
In recent years, Internet traffic has increased explosively due to the increased use of smart-phones, tablet computers, etc. This causes a shortage problem of wireless spectrum. 
Cognitive wireless is considered as a promising solution to this problem~\cite{3,Mitola,Wang10,Akyildiz08,Letaief09}. For recent development of cognitive radio networks, we refer to the survey paper by Ostovar et al. ~\cite{Ostovar20}. In wireless networks, secondary users (unlicensed users) are allowed to cognitively use the bandwidths that are originally allocated to primary users (licensed users). Secondary users should use the bandwidths in such a way that does not interfere primary users. 
In particular, secondary users can use the bandwidths only if primary users are not present.  To this end, secondary users must be aware of the presence of primary users so as to evacuate upon arrivals of primary users. From this point of view, primary users have absolute priority over secondary users meaning that the transmission of a secondary user might be interrupted by a primary user. We assume that interrupted secondary users evacuate to the head of the buffer and {\bf resume} their transmission as soon as a channel is available. 

Motivated by the above situation, we propose analyzing a multiserver queueing system with an infinite buffer for secondary users while primary users have absolute priority over secondary users and are lost if all channels are already occupied by other primary users. Under this assumption (and under Poisson inputs), from the view point of primary users, the system of servers behaves as an Erlang loss system while from that of secondary users, the system is an infinite buffer model where secondary users are served when some servers are not occupied by primary users. In our model, the service times of primary and secondary customers follow two distinct arbitrary distributions. 

A closely related model is the paper by Mitrani and Avi-Itzhak~\cite{Mitrany68}. In this paper, the author considers an M/M/$N$ system where each server is subject to random breakdowns and repairs. The author analyzes the joint distribution of the queue length and the state of the servers using a generating function approach. The model in~\cite{Mitrany68} can be considered as a model with $N$ primary customers. In our model, primary customers arrive according to a Poisson process implying that there are infinite number of primary customers. Akutsu and Phung-Duc~\cite{Akutsu19} examine a closely related Markovian  model 
in which  secondary customers first sense the channels before occupying them.
The stability condition is conjectured and verified by simulation in~\cite{Akutsu19}.
Salameh et al. \cite{Salame17,Salameh20} consider models with limited number of sensing secondary users. Other queueing models of cognitive radio networks could be found in~\cite{wong,konishi,osama,Shajin19,Zhang19,Dudin16,Paul18}. In~\cite{konishi,osama,Dudin16,Liu19}, the service time distributions of primary and secondary customers are either restricted to Markovian distributions (exponential or phase-type distributions) and/or the assumption that the number of active secondary users are finite. In~\cite{Shajin19,Zhang19,Nazarov18,Paul18,Dimitriou,Dragieva,Kim12,Azarfar14}, models with single channel are investigated. 

In this paper, we first relax these assumptions by considering arbitrary distributions for service time of primary users and secondary users. Under these assumptions, we are able to obtain an explicit stability condition. To the best of our knowledge, this is the first analytical result for cognitive radio network models with multiple channels. 
Next, assuming exponential service time distributions, we study the model in depth using the matrix analytic method\cite{Neuts}. This allows us not only to recover the stability condition but also to numerically evaluate effects of the input parameters on the performance of secondary users. The advantage of the matrix analytic method is to provide a systematic way to analyze more complicated models with finite buffers for primary customers and/or Markovian arrival processes (MAP) as well as Phase-type (PH) distributions~\cite{Dudin16}. The  priority queues 
dominate among 
models for service differentiation in communication networks and service systems, and by this reason applications of our model are not restricted to cognitive wireless 
networks~\cite{Vishnevskiy20}. 

The rest of this  paper is organized as follows. In section~\ref{description}, we describe our model in detail 
with focus on the regenerative setting.
In section~\ref{stability_analysis}, we present the stability analysis of  the basic  model, while in section~\ref{multiclass}, as a by product of our analysis, we derive the stability condition for multiserver multiclass system without losses. Section~\ref{Markovian_model} presents a detailed analysis for a special case with exponential distributions for which we are able to obtain performance measures. Finally, in Section~\ref{sim_ana}, we present numerical and simulation experiments to show insights into the performance of our system and the sensitivity of service time distribution while  Section~\ref{conclusion:sec} concludes our paper.
 
\section{Description of the system
} \label{description}

We consider the following  modification of the  Erlang system with two classes of customers with $c$ identical servers, Poisson inputs with rates $\lambda_i$, general independent and identically distributed (i.i.d.) 
service times
$\{S_n^{(i)},\,n\ge1\}$ for class-$i$ customers, 
$i=1,2$.
Class-1  customers  have {\it preemptive priority} and  are lost  if meeting all busy servers, while class-2 (non-priority)  customers  stay in the system regardless of the state of the system
and wait in the queue, if any, according to the {\it non-idling}  FCFS (first-come-first-served) service discipline. Class-1 and  class-2 customers correspond to primary and secondary customers while a server is a channel, in cognitive radio networks.
Denote 
service rates
$\mu_i = 1/\E S^{(i)}$,
where $S^{(i)}$ is the (generic) service time of class-i customers. We also denote by  $\{t_n,\,n\ge1\}$ the   arrival instants of the superposed (Poisson) input with rate $\lambda=\lambda_1+\lambda_2$, and i.i.d. exponential interarrival times  $\{\tau_n=t_{n+1}-t_n\}$ 
with  generic interarrival time $\tau$.
By preemptive-resume priority,
a class-1 customer occupies a server  busy by a class-2 customer, provided there are no idle servers upon his arrival.  

 We deduce stability conditions of this system  based on the regenerative approach.
More precisely, we show that  the basic processes describing the dynamics of the system are  regenerative and then find conditions 
under which these processes are {\it positive recurrent} \cite{Asmus,Smith}. 
This approach, being quite intuitive and transparent, is a powerful tool of  performance and stability analysis of a wide class of queueing processes including non-Markov ones as well.   The main idea of the regenerative stability analysis is based on a characterization of the limiting remaining regeneration time.
More precisely,
if we can show that this time does not go to infinity in probability, then the mean regeneration period length is finite, and thus the process is positive recurrent \cite{Morozov04,Rosario}.


First we describe the regenerative structure of the system. We denote by $Q_i(t)$ the number of class-$i$ customers at time instant  $t^-$ and let $Q(t)=Q_1(t)+Q_2(t)$. Obviously 
\begin{alignat}{1}
&Q_1(t) \le c. 
\label{eq:3}
\end{alignat}
Also let $W_i(t)$ be the workload (remaining work) at instant  $t^-$ of class-$i$ customers, $i=1,2$, and let
 $W(t)=W_1(t)+W_2(t)$. Denote by $S_{i}(t)$  the remaining service time in server $i$
 ($S_i(t)=0$ if the server is idle), and let $R(t)$ be  the  set of servers occupied by class-1 customers at the instant $t$ (we put $R(t)=\emptyset$  if there are no such customers). 
Then 
\begin{alignat}{1}
W_1(t) = \sum_{i \in R(t)}{S_{i}(t)}. \label{eq:4}
\end{alignat}
Denote 

\begin{alignat}{1}
W_n = W(t_n), \quad Q_n = Q(t_n),\nonumber 
\end{alignat}
so, at the arrival instant of  customer $n$  the remaining work in all servers equals $W_n$ and the total number of customers equals $Q_n$.
Note that 

$$
\{W_n=0\} = \{Q_n=0\},\,\,\,\,n\ge 1.
$$ 
Then the   {\it regeneration instants} of  the 
processes $\{Q(t)\}$ and $\{W(t)\}$ are  defined   as  follows

\begin{alignat}{1}
T_{n+1} = \min \{t_k>T_n \, : \, W_k=Q_k=0 \,\}, \, \,\, n \ge 0,\, \,\,T_0:=0, 
\label{eq:7}
\end{alignat}
with generic regeneration period length $T$. It means that $T$ is distributed as any  distance between  two regeneration points, i.e. as $T_{n+1}-T_n$.  (To define recursion (\ref{eq:7}) for $n\ge1$, we put $T_0=0$ by definition.) 
To explain, we note that 
 a regeneration happens when a customer (of any class) meets a completely idle system, and in this case the random  instant   $T_n$  is the arrival instant  of  the $n$th customer who  meets an idle system.

In what follows we assume that 
the first primary customer (of any class) arrives in the idle system at instant $t_1=T_0=0$, in which case $t=0$ is indeed the first (initial)  regeneration point.
It is called {\it zero initial state}.
We call regenerative process  (in continuous-time case) {\it positive recurrent} if the mean regeneration period length is finite, that is  $\E T<\infty.$
A key observation  is that, because of the  priority, the processes $\{Q_1(t)\}$, $\{W_1(t)\}$ describing  class-1 customers  
are regenerative with regeneration instants 

$$
T^{(1)}_{n+1}=\{t_n^{(1)}>T_n^{(1)}: Q_1(t_n^{(1)})=0\},\,n\ge 0  \,\,(T^{(1)}_{0}=0),
$$
with generic regeneration period length $T^{(1)}$, where $\{t_n^{(1)}\}$  are the arrival instants of class-1 customers.  In other words, $T_n^{(1)}$ is the $n$th arrival instant of class-1 customer which meets no other class-1 customers in the system.
Because of  \eqref{eq:3}, the process $\{Q_1(t),\,t\ge0\}$ is {\it tight},
that is, for any $\varepsilon>0$, there exists a constant $C$ such that

$$
\inf_{t\ge0}\P(Q_1(t)\le C)\ge 1-\varepsilon.  
$$
Then  it follows from (\ref{eq:4}) and from the tightness of the process $\{S_i(t),\,t\ge0\}$  that the process $\{W_1(t),\,t\ge0\}$ is tight as well, see  \cite{Morozov97,Morozov04}. 
It then follows from  \cite{Rosario} that these processes (describing  class-1 customers solely) are  also   positive recurrent, that is 

\begin{alignat}{1}
\E \, T^{(1)} < \infty. \nonumber
\end{alignat}
Define the {\it remaining regeneration time} in the system at instant $t$
as

\begin{alignat}{1}
T(t) = \min_{k}{(T_k-t: T_k-t\ge0)},\,\,t\ge0.
\label{7} 
\end{alignat}
In the regenerative stability analysis below we use  the following basic result from the {\it renewal theory} \cite{Feller,Rosario}: 
if there exist a non-random sequence of time instances $z_i\to \infty$, as $i\to \infty$,  and constants $\delta>0$ and $D<\infty$  such that

\begin{eqnarray}
\inf_i\P(T(z_i)\le D)\ge \delta,\label{9a}
\end{eqnarray}
then $\E T<\infty$. 

\section{Stability analysis}\label{stability_analysis}
We will prove conditions which imply positive recurrence of the system. On the other hand, because the processes related to class-1 customers are positive recurrent and the  inputs are Poisson,
in particular, there exists a stationary probability  $\mathsf{P}_i$  that exactly $i$ servers are busy by class-1 customers. In other words, there exist the limits

\begin{eqnarray}
\lim_{t\to\infty}\P(Q_1(t)=i)=\P_i,\,\,i=0,\ldots,c.
\label{9}
\end{eqnarray}
Also denote $\rho_2=\lambda_2/\mu_2.$ It is assumed that the first customer arrives in an empty system at instant $t_1=0$, and if the system is empty at this instant, we call it {\it zero initial state} or initially empty system.  Note that in this case  the instant $T_0=0$ is indeed a regeneration point and then  the first regeneration period is {\it stochastically} equivalent to generic period, that is 
$T_1=_{st} T$.
In this case  positive recurrence means that $\mathsf{E}T<\infty$.
In this section we prove the following main result.

\begin{theorem}\label{Th1}
If condition

\begin{alignat}{1}
&\rho_{2} + \sum_{i=1}^{c}{i \mathsf{P}_{i}} < c,
\label{eq:2:stat_cond}
\end{alignat}
holds then $\mathsf{E}T<\infty$, that is initially empty system is 
positive recurrent.
\end{theorem}
\begin{proof}
 In the interval of time $[0,t)$, denote:  $\hat{V}_1(t)$  the work of class-1 customers accepted by the system (which means $\hat{V}_1(t)$ does not include the lost work); $V_2(t)$  the received work of class-2 customers; $B(t)$  the aggregated busy time of the servers which  equals  departed  work in $[0,t)$. Also denote  $I(t)=\sum_{i=1}^{c}{I_{i}(t)}$, where $I_{i}(t)$ is the idle time of server $i$ in $[0,\,t]$, so $B(t)=ct-I(t)$.
It is now clear that the following balance equation holds.

\begin{eqnarray}
\hat{V}_1(t) + V_2(t) = W_1(t) + W_2(t) + B(t) = W_1(t) +W_2(t)+ ct - I(t). \label{eq:11}
\end{eqnarray}
First of all we note that, by   positive recurrence of "class-1 processes",  
we obtain  (see \cite{Smith}) that

\begin{alignat}{1}
\lim_{t\to\infty}\frac{1}{t}\E Q_1(t) = 0, \quad
\lim_{t\to\infty}\frac{1}{t}\E W_1(t) = 0.  
\label{eq:9}
\end{alignat}
Note that the  convergence with probability 1 (w.p.1),  that is,

$$
Q_1(t)=o(t),\,\,\,W_1(t)=o(t),\,\,t\to \infty,
$$ 
also holds
 in (\ref{eq:9}). 
Now we apply a  proof by contradiction and  assume that, under condition  (\ref{eq:2:stat_cond}),

\begin{alignat}{1}
Q_2(t) \Rightarrow \infty,\,\,\,t\to\infty, \label{eq:10}
\end{alignat}
which means the queue size of class-2 customers  {\it increases infinitely  in probability}. 
It means that, for each fixed $k$,

\begin{eqnarray}
 \lim_{t\to\infty}\mathsf{P}(Q_2(t)>k)= 1.
\label{14}
\end{eqnarray}
 Denote by $\1(\cdot)$ the  indicator function, then 
 
$$
I_i(t)=\int_0^t\1 (S_i(u)=0)du,\,\,\,t\ge0,
$$
is the idle time of server $i$ in the interval $[0,\,t]$. It then  follows that the average aggregated  idle time of all servers 
in the interval of time $[0,\,t]$ is
\begin{eqnarray}
\E I(t) =\E \Big[\sum_{i=1}^c\int_0^t\1(S_i(u)=0)du\Big]
= \sum_{i=1}^c\int_0^t\P(S_i(u)=0)du. \nonumber
\end{eqnarray}
Because the service discipline is non-idling, then
\begin{eqnarray}
\mathsf{P} (Q_2(t) > c) \le \mathsf{P} (S_{i}(t) > 0),\,\,\,i=1,\ldots,c, 
\label{eq:16}
\end{eqnarray}
and, by \eqref{14}, (\ref{eq:16}), we obtain as $t\to\infty$,

\begin{alignat}{1}
\P(S_{i}(t)=0)=1-\P(S_{i}(t)>0)\le 1- \P(Q_2(t) > c) \rightarrow 0. \nonumber
\end{alignat}
 It  now  easily follows  that 
 
\begin{alignat}{1}
\lim_{t\to\infty}\frac{1}{t}\E I(t)=\lim_{t\to\infty}\sum_{i=1}^c\frac{1}{t}\int_0^t\P(S_i(u)=0)du =0. 
\label{eq:19}
\end{alignat}
Denote by $A_2(t)$  the number of class-2 customers arrived  in the interval $[0,t)$. The process  $\{A_2(t),\,t\ge0\}$ is a  positive recurrent {\it cumulative process} \cite{Smith}  
or, equivalently,  the process with {\it regenerative increments} \cite{Serfozo} in which each arrival instant of class-2 customer is  a regeneration instant of the incoming load process  

$$
V_2(t)=\sum_{k=1}^{A_2(t)}S_{k}^{(2)},\,\,t\ge0.
$$
We note that the i.i.d. exponential interarrival times $\{\tau_n^{(2)}\}$ between class-2 customers are the {\it regeneration periods} of the load process $\{V_2(t),\,t\ge0\}$. Denote   the generic  period $\tau^{(2)}$
which has  parameter $\lambda_2$. Also  we note that the process $\{V_2(t)\}$ has the  increment $S^{(2)}$ over interval $\tau^{(2)}$.  Then it  follows, for instance, from Theorem  55 in  \cite{Serfozo},  that  

\begin{eqnarray}
\lim_{t\to\infty}\frac{\mathsf{E}V_2(t)}{t} 
= \frac{\E S^{(2)}}{\E\tau^{(2)}} =\rho_2.
 \label{17}
\end{eqnarray}
 Now we define the total busy time $B_1(t)$ when all servers are occupied by  class-1 customers, in the interval $[0,\,t]$. A key observation is that $B_1(t)$ can be represented as follows
 
\begin{eqnarray}
B_1(t)&=&\int_0^t\1(Q_1(u)=1)du+2\int_0^t\1(Q_1(u)=2)du+\cdots\nonumber\\
&+&c\int_0^t\1(Q_1(u)=c) du
=\sum_{i=1}^c i\int_0^t\1(Q_1(u)=i)du.\nonumber
\end{eqnarray}
This gives the following equality 

\begin{alignat}{1}
\hat{V}_1(t) =B_1(t)+W_1(t)= \sum_{i=1}^ci\int_0^t\1(Q_1(u)=i) du 
+ W_1(t),\,\,\,t\ge0.\nonumber
\end{alignat}
Recall that all processes related to class-1 customers are positive recurrent regenerative.
Moreover, since the input is Poisson,
the weak limit $Q_1(t)\Rightarrow Q_1$ exists  
and is the stationary number of servers occupied by class-1 customers with distribution 
(\ref{9}).
It is worth mentioning that, by the property PASTA \cite{Asmus}, the stationary probability $\P_i$
is also  the {\it limiting fraction of the time} when there are exactly $i$ class-1 customers  in the system, which means

\begin{alignat}{1}
\mathsf{P}_{i} = \lim_{t\to\infty}\frac{1}{t}\int_{0}^{t}{\P(Q_1(u)=i) du}
,\,\,i=1,\ldots,c.\label{15a} 
\end{alignat}
Now it follows from (\ref{eq:9}), (\ref{20}), (\ref{15a}) that 

\begin{eqnarray}
\lim_{t\to\infty}\frac{1}{t}\E\hat{V}_1(t)= 
\sum_{i=1}^{c}i\lim_{t\rightarrow\infty}\frac{1}{t}\int_{0}^{t}{\mathsf{P}(Q_1(u)=i) du}=\sum_{i=1}^ci\mathsf{P}_i. 
\label{21}
\end{eqnarray}
Now collecting together results (\ref{eq:9}), (\ref{eq:19}), (\ref{17}) and (\ref{21}), we obtain  from the balance equation (\ref{eq:11}) that  

\begin{alignat}{1}
\rho_2+ \sum_{i=1}^{c} i \P_i=  \lim_{t\to\infty}\frac{1}{t}\E W_2(t)+c, \nonumber
\end{alignat}
implying 

\begin{alignat}{1}
\lim_{t\rightarrow\infty}\frac{1}{t}\E W_2(t) = \sum_{i=1}^{c}i \P_i + \rho_2 - c \ge 0,\nonumber 
\end{alignat}
or 

\begin{alignat}{1}
\rho_2 + \sum_{i=1}^{c}
i \P_i \ge c. \nonumber
\end{alignat}
This contradiction with assumption
(\ref{eq:2:stat_cond})
shows that assumption 
(\ref{eq:10}) is false, and thus 
$Q_2(t)\not \Rightarrow \infty$. 
Hence 
there exist constants $N<\infty,\,\delta_0>0$ and a (non-random) sequence $u_i\to \infty$ such that (cf. (\ref{9a}))

\begin{alignat}{1}
\inf_i\mathsf{P} (Q_2(u_i) \le N) \ge \delta_0.
\label{18}
\end{alignat}
Note that

$$
W_2(t)\le \sum_{k=1}^{Q_2(t)}S_k^{(2)}+\sum_{k=1}^cS_k(t),\,\,\,t\ge0,
$$
and that the remaining service times
$\{S_k(t),\,t\ge0\}$ are the tight processes
$k=1,\ldots,c$, see  \cite{Morozov97}.
A routine but tedious  calculation   confirms an intuitive result, that the bound (\ref{18}) 
implies  the corresponding lower   bound  for the workload process $\{W_2(t),\,t\ge0\}$, namely
\begin{eqnarray}
\inf_i\mathsf{P} (W_2(u_i)\le D_2)\ge \delta_2,
\nonumber
\end{eqnarray}
for some constants $D_2,\,\delta_2>0$.
Moreover, because the positive recurrent process $\{W_1(t),\,t\ge0\}$ is also tight, in particular,  
\begin{eqnarray}
\inf_i\mathsf{P} (W_2(u_i)\le D_2,\,W_1(u_i)\le D_1)\ge \delta_1,\nonumber
\end{eqnarray}
for some constants $D_1,\,\delta_1>0$.
Denote $\tau(t)$ the remaining (exponential) interarrival time at instant $t$ (in the superposed Poisson input with rate $\lambda$), so  
$\P (\tau(t) \le x) \ge e^{-\lambda x}$ for any $x\ge0$. 
Recall definition \eqref{7}  and note that,
for arbitrary fixed $u_i$ (satisfying (\ref{18})) and a constant $L>D_1+D_2$,  the following  lower bound for the remaining regeneration time holds.
\begin{eqnarray}
\nonumber
\P (T(u_i) \le L) &\ge& \P\Big (W_2(u_i) \le D_2,\; W_1(u_i) \le D_1,\, L \ge \tau(u_i) > D_1+D_2\Big) \nonumber \\
&\ge& \delta_1(e^{-\lambda (D_1+D_2)}-e^{-\lambda L})=:\varepsilon>0.\label{19}
\nonumber
\end{eqnarray}
To explain this inequality, we note that on the event 

$$
\Big\{W_2(u_i) \le D_2,\; W_1(u_i) \le D_1,\, L \ge \tau(u_i) > D_1+D_2\Big\},
$$ 
the first customer arriving after instant $u_i$
(at instant $u_i+\tau(u_i)\le u_i+L$) meets a completely idle system and thus a regeneration occurs.  It explains why then $T(u_i) \le L$.
Because this bound  is uniform in $u_i$ and $i$, we  obtain  that  the inequality  (\ref{9a}) holds and thus $\E T<\infty$.
Thus  we conclude that \eqref{eq:2:stat_cond} is the {\it sufficient  stability condition}.
\end{proof}
Now we show that  \eqref{eq:2:stat_cond} is the necessary stability condition
as well. Namely, we assume that the initially empty  system is  positive recurrent, which means $\mathsf{E}T<\infty$.  Note that

$$
I(t)\ge \int_0^t\1(Q(u)=0)du=:I_0(t),
$$
where $I_0(t)$ is the aggregated time when all servers are {\it simultaneously free} in the interval $[0,\,t]$.
Then it follows from the  theory of regenerative processes that 
\begin{eqnarray}
\lim_{t\to\infty}\frac{I(t)}{t}\ge \lim_{t\to\infty}\frac{I_0(t)}{t}
=\frac{\E I_0}{\E T},
\label{28}
\end{eqnarray}
where $I_0$ is the idle time, during  a regeneration period, when  all servers are simultaneously free   \cite{Asmus,Serfozo}.
Because $\tau$ is exponential and $\E S^{(i)}<\infty$, for arbitrary $\delta>0$ and   some $\varepsilon_1>0$,

$$
\mathsf{P}(\tau>\delta +S^{(i)})=\varepsilon_1,\,\,\,i=1,2.
$$ 
Hence, 

$$
\mathsf{E}I_0\ge \mathsf{E}(I_0;\tau-S^{(i)}>\delta)\ge \delta\delta_1>0,
$$
and  thus the limit in (\ref{28}) is positive.
It then immediately follows from  (\ref{17}), (\ref{21}), (\ref{28}), and from the balance equation (\ref{eq:11}),
by dividing both sides by $t$ and letting $t\to\infty$, that the inequality (\ref{eq:2:stat_cond}) holds.  This implies that 
(\ref{eq:2:stat_cond}) is indeed  the necessary stability condition.
Combining this result with the statement of Theorem \ref{Th1} we obtain 
the following main statement.
\begin{theorem}
The initially empty system under consideration is positive recurrent if and only if condition (\ref{eq:2:stat_cond})
 holds.
\end{theorem}

\begin{remark}
Following \cite{Rosario}  one can  extend  the analysis to prove that, under condition (\ref{eq:2:stat_cond}), 
the system is positive recurrent
under arbitrary fixed values 
$W(0)$ and $Q(0)$.
In this case we shall also show  that the 1st regeneration period, which in general has another distribution than $T$, is finite w.p.1.
\end{remark}
\bigskip
It is worth mentioning that, because to analyse  class-1 customers   only, we can treat the system as a loss  $M/G/c/0$ system, then 
the stationary distribution $\{\mathsf{P}_i\}$ can be found from the celebrated {\it Erlang forula}, for instance, see \cite{Asmus}.    

 Denote by $\zeta$  the stationary number of  servers available for class-2 customers, that is (stochastically),
$$
\zeta=_{st}c-\sum_{i=1}^c i\1(Q_1=i).
$$
Now, denoting $\Delta=\E \zeta$, we obtain that  the mean stationary number of servers available for class-2 customers is
$$
\Delta=c- \sum_{i=1}^c i\mathsf{P}_i.
$$
  Then condition  (\ref{eq:2:stat_cond}), rewritten as  
\begin{eqnarray}
    \rho_2 <\Delta, \label{20}
\end{eqnarray}
has the following  probabilistic interpretation: the traffic intensity of class-2 customers 
must be less than the mean number of
 the available  servers $\Delta$. 
 Condition (\ref{20}) is a {\it negative drift condition} and,  being  intuitive,  however requires a strict proof because, unlike classic systems with a fixed number of servers, in this case the number of available servers $\zeta$ is {\it random}. 
 This is the main reason why we present above the detailed  regenerative proof of this result.

\subsection{\texorpdfstring{$M/G/c$}{TEXT} system without losses}\label{multiclass}
In this section, keeping previous notation,  we obtain as a by-product of our analysis, a well-known stability condition of a buffered  multiclass system,  that is the system with the {\it infinite capacity buffer for both classes of customers}. Thus the only difference between the original and the new systems is that there are no losses in the new system.

\begin{theorem}
The initially empty system without losses is positive recurrent if and only if the condition

$$
\rho_1+\rho_2
< c
$$
holds.
\end{theorem}
The proof of this statement is quite similar to that has been given  in the previous section. The only difference is that now we replace $\hat V_1(t)$  by the full work $V_1(t)$ generated by all class-1 customers which arrive in the interval of time $[0,\,t)$.  It remains to note  that 

$$
\lim_{t\to\infty}\frac{\E V_1(t)}{t}=\rho_1. 
$$
\begin{remark} Using  previous arguments,  it is straightforward to show that condition 

$$
\sum_{i=1}^{N}{\rho_{i}} < c
$$
is stability criterion 
for $N$-class system with no losses and arbitrary work-conserving discipline including various priority policies.

\end{remark}

\section{Exponential service time distributions}\label{Markovian_model}
In this section, we examine more closely a special case of the system described in Section 2, in which  service times of both  customer classes follow exponential distributions, with parameter $\mu_i$ for class-$i$ customers. 

In this pure Markovian case, for stability analysis  it is possible by applying the {\it matrix analytic method}  \cite{Neuts}.
This alternative proof of stability is instructive. Moreover, in this setting we can calculate the stationary distribution of the basic Markov process and as a result obtain the corresponding stationary performance indexes. In this setting, the process  $\{(Q_1(t), Q_2(t)),\,t \ge 0 \}$ is a continuous-time Markov Chain with the state space $\mathbb{S}$ given by

$$
\mathbb{S}=\{(i,j) \in \{0,1,...,c\}\times \mathbb{N}\}.
$$
Grouping states into levels according to their values of $Q_2(t)$, the system can be formulated as a Quasi-Birth-and-Death  process (QBD) whose the infinitesimal generator $Q$ is expressed as follows	

$$
Q =
\begin{pmatrix}
B_0 & C & O & O & O & ... & ...&...&...\\
A_{1} & B_1 & C & O & O & ... & \vdots&\vdots&\vdots\\
O & A_{2} &B_2 & C & O  & ... &\vdots&\vdots&\vdots\\
\vdots&\vdots& \ddots &\ddots & \ddots&\vdots&\vdots&\vdots&\vdots \\
O & O & ... & A_{c} &B_c & C & O &...&\vdots \\
O & O & ... & O & A_{c} &B_c & C & O &... \\
\vdots&\vdots&\vdots &\vdots & \vdots&\ddots &\ddots & \ddots & \vdots& \\
\end{pmatrix},
$$
where $O$ is a $(c+1) \times (c+1)$ zero matrix, and $A_n, B_n, C$ are $(c+1) \times (c+1)$ block matrices given by

$$
C =
\begin{pmatrix}
\lambda_2 & 0 & 0 & 0  \\
0 & \lambda_2 & 0 & 0  \\
\vdots& \vdots  &\ddots &\vdots  \\
0 & 0  & 0 & \lambda_2
\end{pmatrix},
$$

$$
A_n =
\begin{pmatrix}
n\mu_2 & 0 & 0 & ... & 0 & 0 &...&0\\
0 & n\mu_2 & 0 & ... & 0 & 0 &...&0\\
\vdots &\vdots  & \ddots & \vdots & \vdots & \vdots & \vdots &\vdots\\
0 & 0 & ... & n\mu_2&0&...&0&0\\
0 & 0 & ... & 0 & (n-1)\mu_2 & ... &0&0\\
\vdots & \vdots & \vdots & \vdots & \vdots & \ddots & \vdots & \vdots\\
0 & 0 & ... & 0 & 0 & ...& \mu_2 & 0\\
0 & 0 & ... & 0 & 0 & ... & 0 & 0
\end{pmatrix},
$$
for $n \le c$, and $A_n = A_c$, for $n > c$;

$$
B_n =
\begin{pmatrix}
b_{n,0} & \lambda_1 & 0 & ... & 0 & 0 &...&0&0\\
\mu_1 & b_{n,1} &\lambda_1 & ... & 0 & 0 &...&0&0\\
0 & 2\mu_1 & b_{n,2} & \ddots & 0 & 0 &...&0&0\\
\vdots &\vdots  & \ddots & \ddots & \ddots & \vdots & \vdots &\vdots &\vdots\\
0 & 0 & 0 & ... & b_{n,c-n}&\lambda_1&...&0&0\\
0 & 0 & 0 & ... & (n-1)\mu_1 & b_{n,c-(n-1)}& \ddots &0&0\\
\vdots & \vdots & \vdots & \vdots & \vdots & \ddots & \ddots & \ddots & \vdots\\
0 & 0 & 0 & ... & 0 & 0 & ...& b_{n,c-1} & \lambda_1\\
0 & 0 & 0 & ... & 0 & 0 & ... & c\mu_1 & -(\mathrel{{\lambda_2}{+}{c\mu_1}})\\
\end{pmatrix},
$$	
for $n \le c$, and $B_n = B_c$ for $n > c$.
In case $n \le c$, the diagonal elements of $B_n$ are given by	

$$
b_{n,i}=-(\lambda_1+\lambda_2+i\mu_1+min\{c-i,n\}\mu_2).
$$



Denote by $\boldsymbol{\eta} = (\eta_0, \eta_1,...,\eta_c)$ the row vector representing the stationary distribution of the infinitesimal generator $Q^*=A_c+B_c+C$. It is easy to see that 



$$
\boldsymbol{\eta}=\left(\eta_0, \rho_1 \eta_0, \frac{\rho_1^2}{2!}\eta_0,...,\frac{\rho_1^i}{i!}\eta_0,...,\frac{\rho_1^c}{c!}\eta_0\right), 
$$
where 

$$
\eta_0=\frac{1}{\sum_{i=0}^{c} \frac{\rho_1^i}{i!}}\,\,\, \mbox{and}\,\,\, \rho_1=\frac{\lambda_1}{\mu_1}.
$$
The Quasi-Birth-Death process is ergodic if and only if the following condition holds~\cite{Neuts}.

$$
\boldsymbol{\eta} C\boldsymbol{e} < \boldsymbol{\eta} A_c\boldsymbol{e}, 
$$
where $\boldsymbol{e}$ denotes the $(c+1)$-dimension column vector of ones. This stability condition can be further transformed as 

\begin{align}
\lambda_2 < \frac{\sum_{i=0}^{c-1}(c-i)\frac{\rho_1^i}{i!}}{\sum_{i=0}^{c}\frac{\rho_1^i}{i!}}\mu_2
. \label{eq:39}
\end{align}
The above condition can also be rewritten as

$$
\rho_2 + \sum_{i=1}^{c}i\eta_i<c,
$$
which is identical to the stability condition in general case stated in \eqref{eq:2:stat_cond}.


In what follows, we consider the system under the stability condition. Now, let $\pi(i,j) = \P(Q_1(t)=i,Q_2(t)=j)$ denote the stationary probability of the Markov chain, and let 

$$
\boldsymbol{\pi}_{j} = (\pi(0,j),\pi(1,j),...,\pi(c,j))\,\,\,
\mbox{for}\,\,\, j \in \mathbb{N}.
$$
According to {\it Matrix-analytic-method}~\cite{Neuts,Phung-Duc10}, we have

$$
\boldsymbol{\pi}_{j}=\boldsymbol{\pi}_{c}R^{j-c}, \quad j \geq c, \qquad \boldsymbol{\pi}_{j} = \boldsymbol{\pi}_{j-1}R^{(j)}, \quad j =c,c-1,\dots,1,
$$
where $R$ is the minimal non-negative solution of 

$$
C + RB_c + R^2A_c = O, 
$$
and 

$$
R^{(j)} = -C(B_j+R^{(j+1)}A_{j+1})^{-1}, \quad j = c-1,c-2,\dots,1,
$$
given that $R^{(c)} = R$ and $R$ is numerically computed using algorithms in~\cite{Neuts}.




Finally, $\boldsymbol{\pi}_{0}$ is the unique solution of the following equations

\begin{align*}
\boldsymbol{\pi}_{0}(B_0+R^{(1)}A_1)& =\boldsymbol{0}, \\
\boldsymbol{\pi}_0\left(I+ \sum_{i=1}^{c-1}\prod_{j=1}^{i}R^{(j)}+\left(\prod_{j=1}^{c}R^{(j)}\right)(I-R)^{-1}\right) \boldsymbol{e} & =1,
\end{align*}
where $I$ denotes the $(c+1)\times(c+1)$ identity matrix.





Moreover, due to Little's law, we  obtain the  average waiting time, $\mathsf{E} W_q$, of class-2 customers as follows

\begin{align*}
\mathsf{E} W_q&= \frac{\sum_{j=0}^{\infty}\sum_{i=0}^{c} \max\{0,i+j-c\}\pi(i,j)}{\lambda_2} \\ \nonumber 
& = \frac{\sum_{j=1}^{c-1}\sum_{i=c-j+1}^{c} (i+j-c)\pi(i,j) + \boldsymbol{\pi}_{c}(I-R)^{-2}R\boldsymbol{e} + \boldsymbol{\pi}_{c}(I-R)^{-1}\boldsymbol{f}}{\lambda_2}, 
\end{align*}
where $\boldsymbol{f}=(0,1,2,...,c-1,c)^T$.






There are other performance measures which describe QoS of such a system. For instance, if a class-2  customer is interrupted by class-1 customer, we call it  a {\it  termination event}. 
Denote the set 

$$
\mathbb{S}^{*} =\{(i,j): i+j \ge c,\,i \le c-1\},
$$
which contains the states when 
all servers are occupied and there is at least one class-2 customer occupying a server. Then the  average number of termination events per class-2 customer 
is given by

$$
\mathsf{E}N_T 
= \frac{\lambda_1}{\lambda_2}\sum_{(i,j)\in \mathbb{S}^{*}} \pi(i,j),
$$
where we apply the  property  PASTA \cite{Asmus}.
We also present a new alternative way how to calculate $\mathsf{E}N_T$ based on the regenerative arguments.
Indeed, denote by $A_1(t)$ the number of class-1 customers arriving
in the interval $[0,\,t)$ and recall similar definition $A_2(t)$ for class-2 arrivals. 
Recall notation 

$$
Q(t)=Q_1(t)+Q_2(t)\Rightarrow Q=Q_1+Q_2.
$$
Then using the basic asymptotic results for positive recurrent  regenerative processes  \cite{Asmus,Rosario},
we obtain $\mathsf{E}N_T$  as the following w.p.1 limit  of fraction of class-1 arrivals  interrupting class-2 customers: 
\begin{eqnarray}
\mathsf{E}N_T&=&\lim_{t\to\infty}\frac{1}{A_2(t)}\sum _{n=1}^{A_1(t)}\1(Q(t_n^{(1)})\ge c,\,Q_1(t_n^{(1)})\le c-1)\nonumber\\
&=&\lim_{t\to\infty}\frac{A_1(t)}{t}\frac{1}{A_1(t)}\sum _{n=1}^{A_1(t)}\1(Q(t_n^{(1)})\ge c,\,Q_1(t_n^{(1)})\le c-1)
\cdot\lim_{t\to\infty}\frac{t}{A_2(t)}\nonumber\\
&=&\lambda_1\mathsf{P}(Q\ge c,\,Q_1\le c-1) \frac{1}{\lambda_2}=\frac{\lambda_1}{\lambda_2}\sum_{(i,j)\in \mathbb{S}^{*}}
\mathsf{P}(Q_1+Q_2=i+j,\,Q_2=j)\nonumber\\
&=&
\frac{\lambda_1}{\lambda_2}\sum_{(i,j)\in \mathbb{S}^{*}} \pi(i,j),\label{42}
\end{eqnarray}
where we also use property PASTA to apply the equality

$$
\pi(i,j)=
\lim_{n\to \infty} \frac{1}{n}\sum_{k=1}^n
\sum_{(i,j)\in \mathbb{S}^{*}}   
\1(Q(t_k^{(1)})=i+j,\,Q_2(t_k^{(1)})=j).
$$
The equality (\ref{42})  written as 

$$
\lambda_2\mathsf{E} N_T=\lambda_1\sum_{(i,j)\in \mathbb{S}^{*}} \pi(i,j)
$$
is intuitive and establishes a balance  between
the rate of the {\it interrupted} class-2 customers and the rate of {\it interrupting} class-1 customers. 
\begin{remark}
It is worth mentioning that relation (\ref{42}) holds also for general service time distribution of any class of customers, however in this case the stationary distribution $\{\pi(i,j)\}$ of (non-Markovian process $\{(Q_1(t),\,Q_2(t)),\,t\ge0\}$) is not analytically available.
\end{remark}

\section{Simulations and Numerical Insights}\label{sim_ana}

In this section, we present some numerical examples of the results obtained by the matrix analytic method presented in Section~\ref{Markovian_model}. In our experiment, for  fixed $\mu_1=4,\mu_2=20$ and values $c=2,5$,
we  show the changes in the values of some  performance measures against $\lambda_1$ and $\lambda_2$.   Under the same settings, we also carry out simulations and obtain the same results as those obtained by the matrix analytic method. Furthermore, to show the sensitivity of the service time distribution of primary users, we also compare the results by matrix analytic methods with those by simulations where service time of class-1 customers follows Erlang distributions with the shape parameter $r=5,10$.

The scale parameter is chosen such that the mean value remains the same as in the case of exponential distributions. The duration for all experimental simulations is set at $10^6$ time units, which is adequate for the simulation results to converge to their corresponding numerical results. The simulation results in all the figures are represented by the points marked with notation $\mathsf{sim}$, without which the results are understood to be obtained from numerical calculations. Simulation results with service time of class-1 customers following Erlang distributions are marked with the abbreviation $\mathsf{Erl.}$ in the legend.

\begin{figure}[H]
\centering
\includegraphics[width=0.8\linewidth]{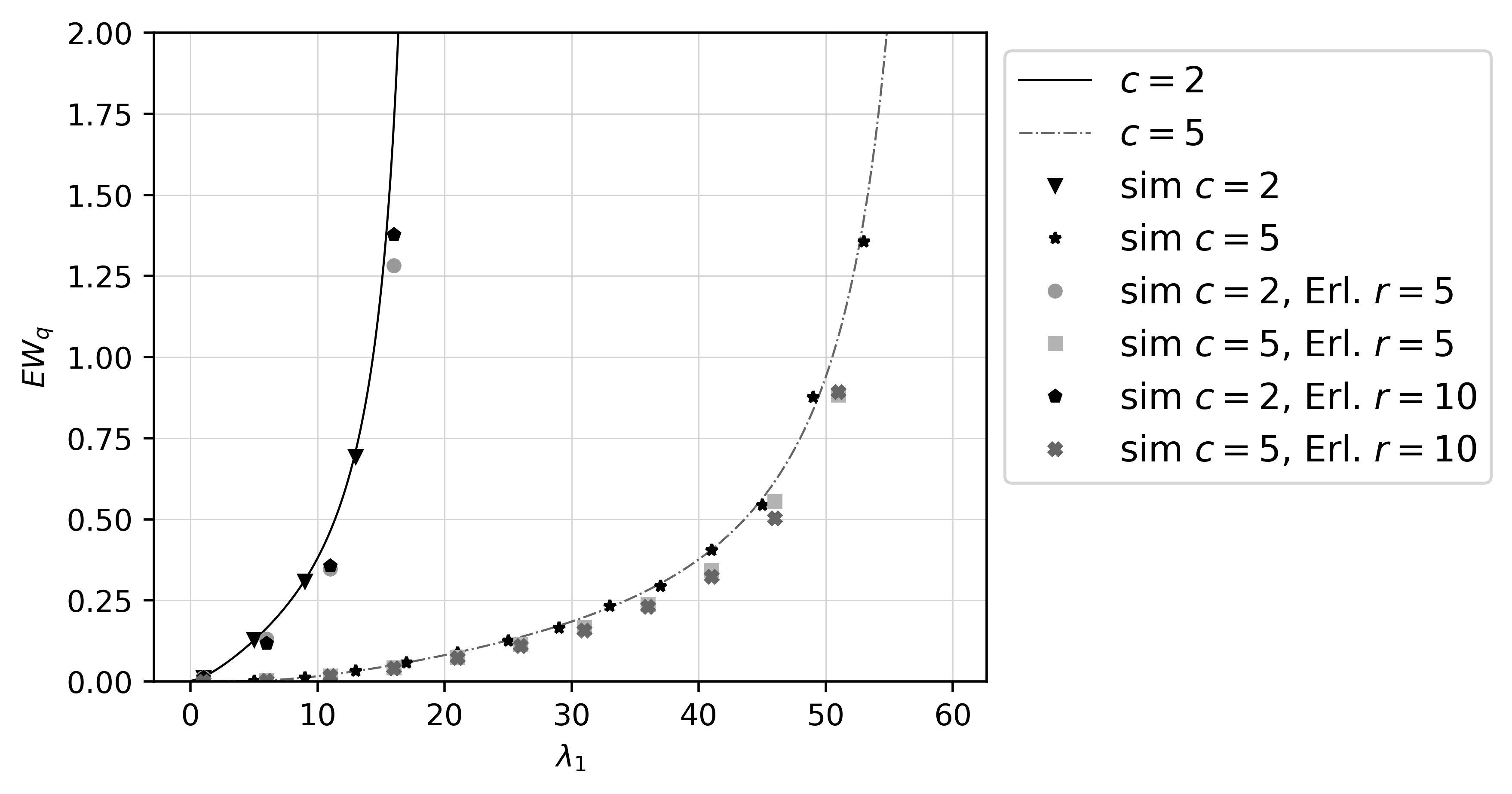} 
\caption{Average waiting time of class-2 customers against arrival rate of class-1 customers $(\lambda_{2} = 8)$.}
\label{fig:1}
\end{figure}
\begin{figure}[H]
\centering
\includegraphics[width=0.8\linewidth]{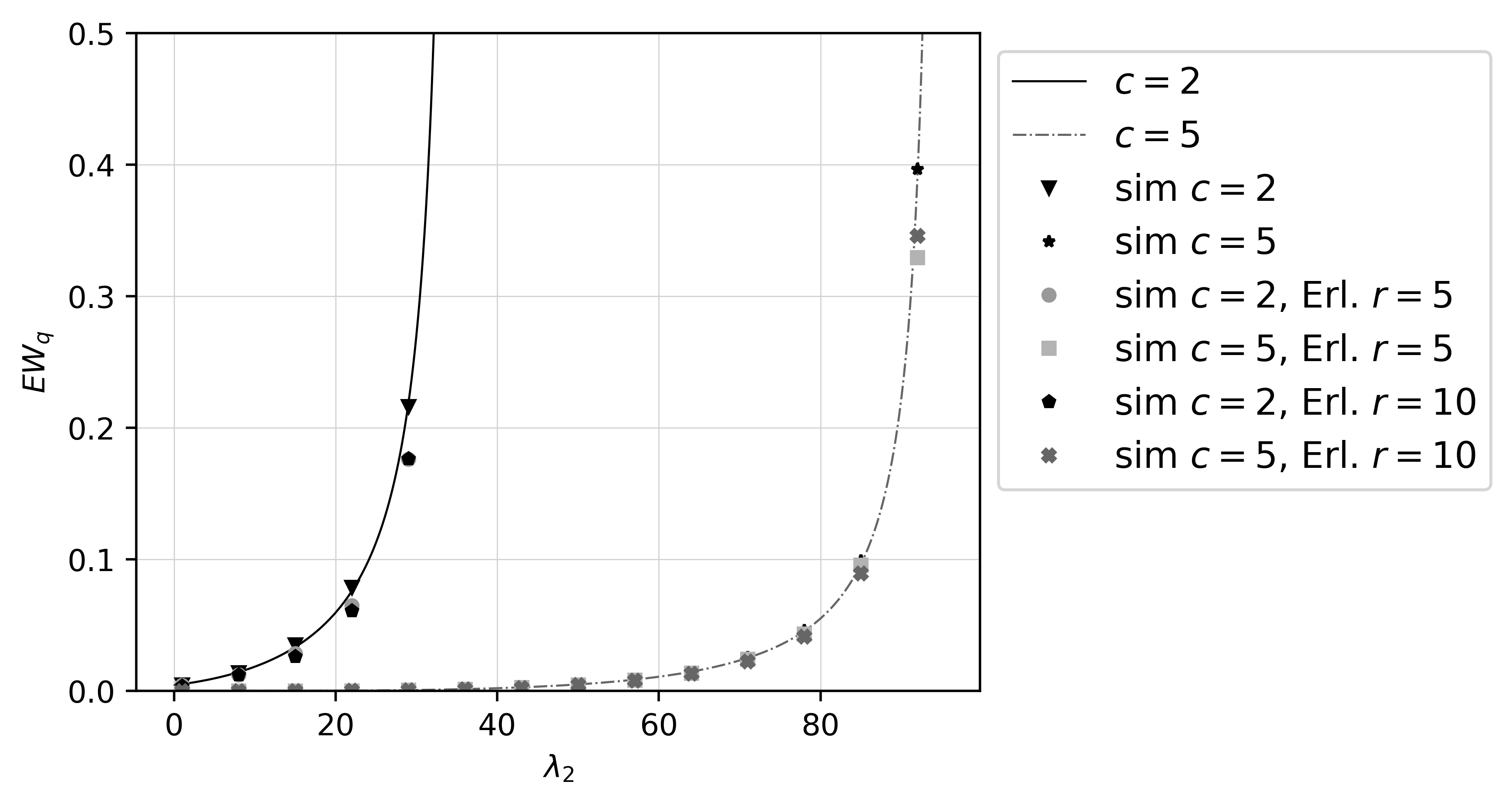} 
\caption{Average waiting time of class-2 customers against arrival rate of class-2 customers ($\lambda_{1} = 1$).}
\label{fig:2}
\end{figure}

Figure \ref{fig:1} and Figure \ref{fig:2} compare average waiting time of class-2 customers as $\lambda_1, \lambda_2$ and $c$ vary. It can be seen that the average amount of time that class-2 customers spend in the system ($\mathsf{E} W_q$) is larger when customers of either class arrive more frequently, or when there are fewer servers. We observe that $\mathsf{E} W_q$ with exponential service time distribution for class-1 is largest while $\mathsf{E} W_q$ with $r=5$ is larger than that with $r=10$. This indicates  that $\mathsf{E} W_q$ increases with the increase in the variance of service time of class-1 customers.

Figure \ref{fig:3} illustrates  how the average number of termination events per class-2 customer ($\mathsf{E}N_T$) changes according to the arrival rates of class-1 customers and the number of servers. Obviously, when class-1 customers arrive more frequently, more class-2 customers' sessions are terminated. Also, as the number of servers is larger, there are more spaces for all customers, and thus class-2 customers are less likely to be kicked out of the servers. The same patterns can be observed in Figure \ref{fig:4}. 
Furthermore, simulations results for Erlang service time distributions of class-1 customers are almost the same as the corresponding ones with exponential distributions. This shows that $\mathsf{E}N_T$ is (almost) insensitive to the service time distribution of class-1 customers. 

\begin{figure}[ht]
\centering
\includegraphics[width=0.8\linewidth]{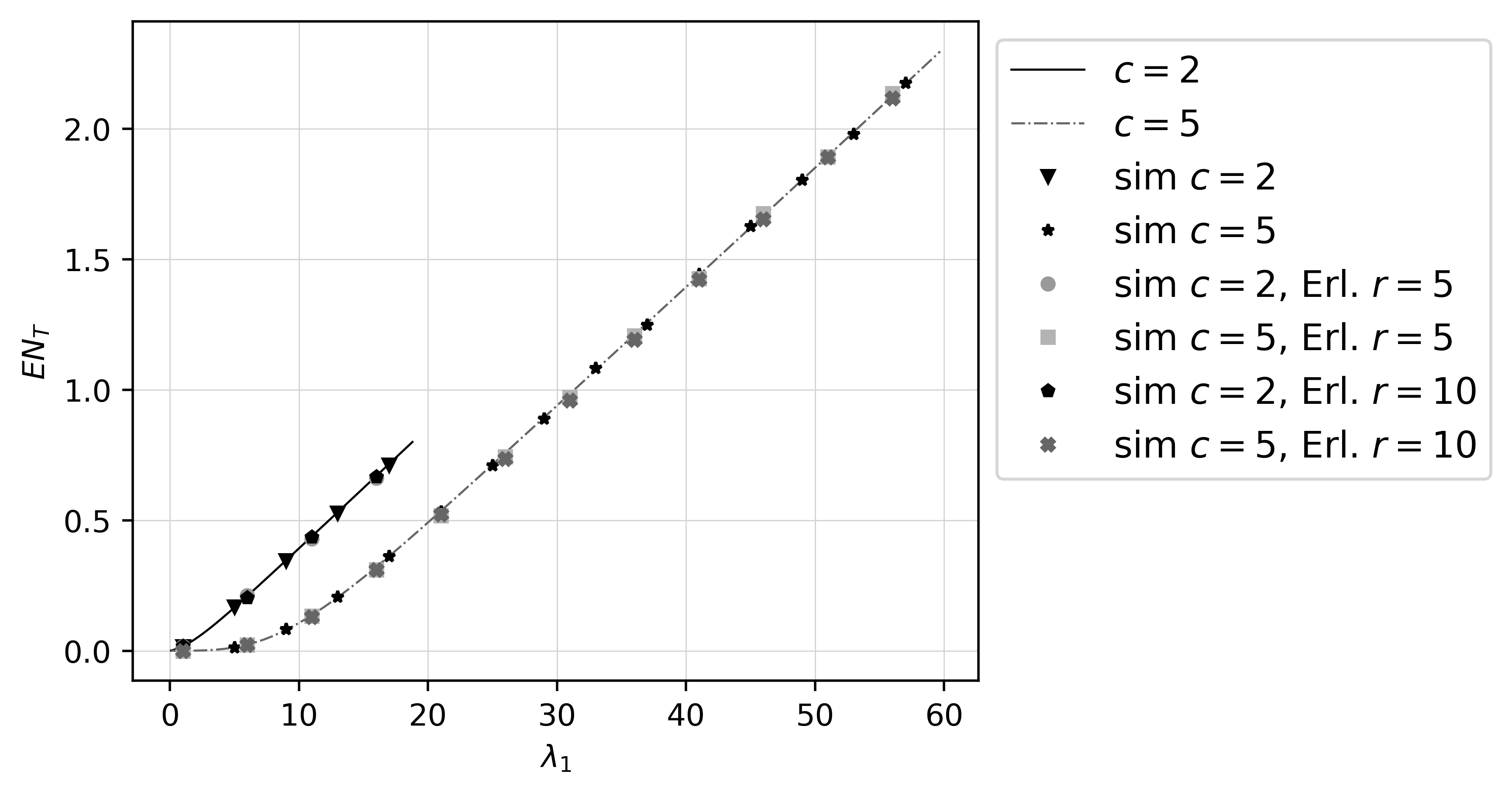} 
\caption{Average number of termination events per class-2 customers against arrival rate of class-1 customers $(\lambda_{2} = 8)$.}
\label{fig:3}
\end{figure}

\begin{figure}[H]
\centering
\includegraphics[width=0.8\linewidth]{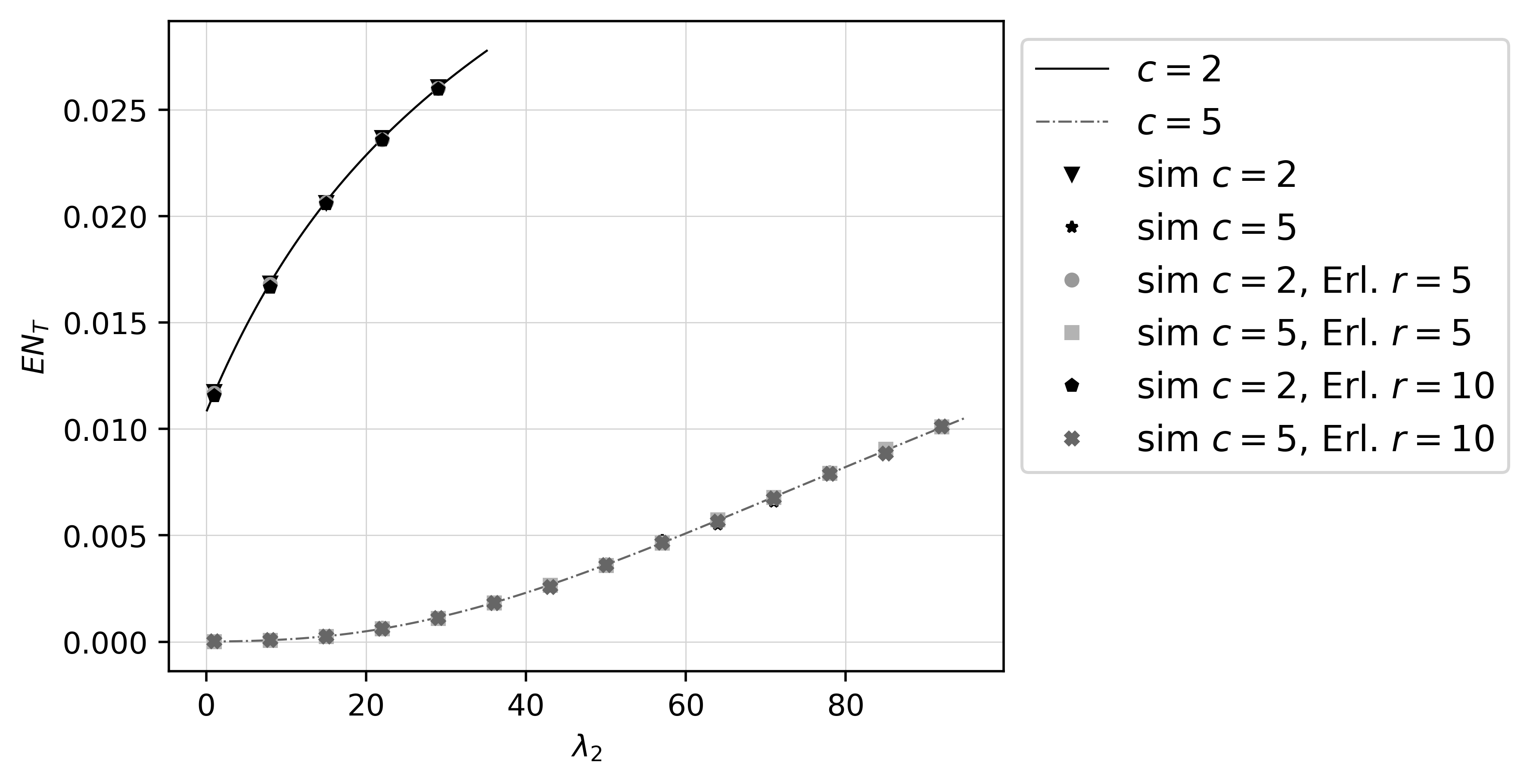} 
\caption{Average number of termination events per class-2 customers against arrival rate of class-2 customers ($\lambda_{1} = 1$).}
\label{fig:4}
\end{figure}

Figure \ref{fig:5} shows the distributions of the number of times that each class-2 customer is terminated. In this experiment, we set $\lambda_1 = 20,\lambda_2 = 30, c = 10$, and let the service time of class-1 customers follow an exponential distribution (denoted by $\mathsf{Exp.}$ in the figure), and Erlang distributions. It can be seen that higher numbers of termination times occur with smaller probabilities. Also, there is no significant difference in the results when we modify the distribution of service time for class-1 customers under the three settings.

\begin{figure}[ht]
\centering
\includegraphics[width=0.6\linewidth]{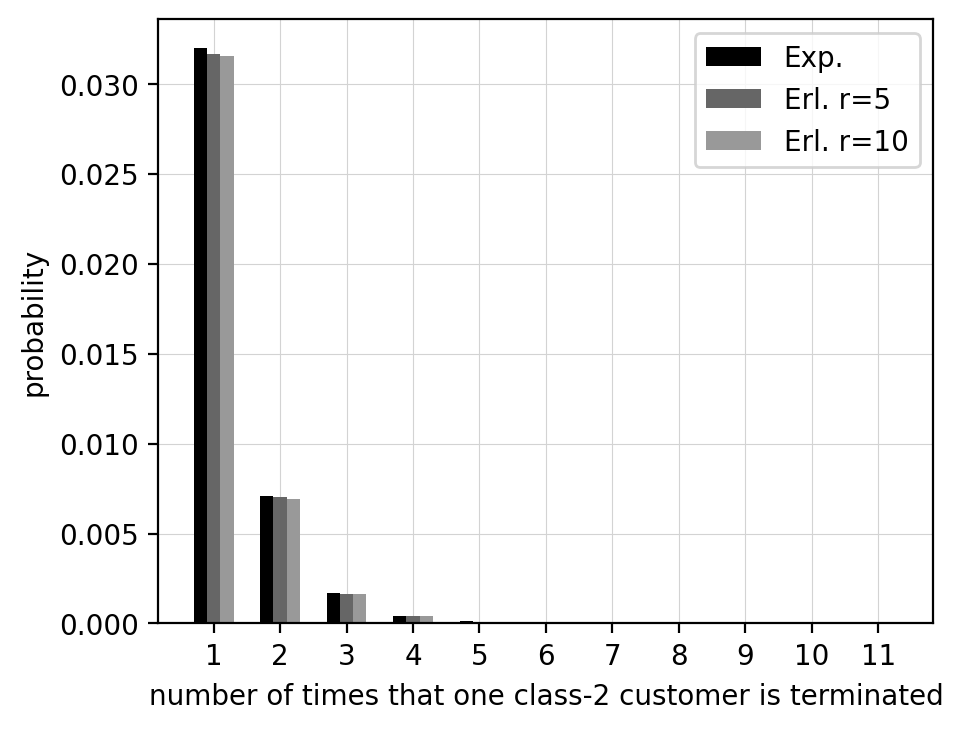} 
\caption{Distribution of number of termination events ($\lambda_1 = 20,\lambda_2 = 30, c = 10$).}
\label{fig:5}
\end{figure}

\begin{figure}[H]
\centering
\includegraphics[width=0.8\linewidth]{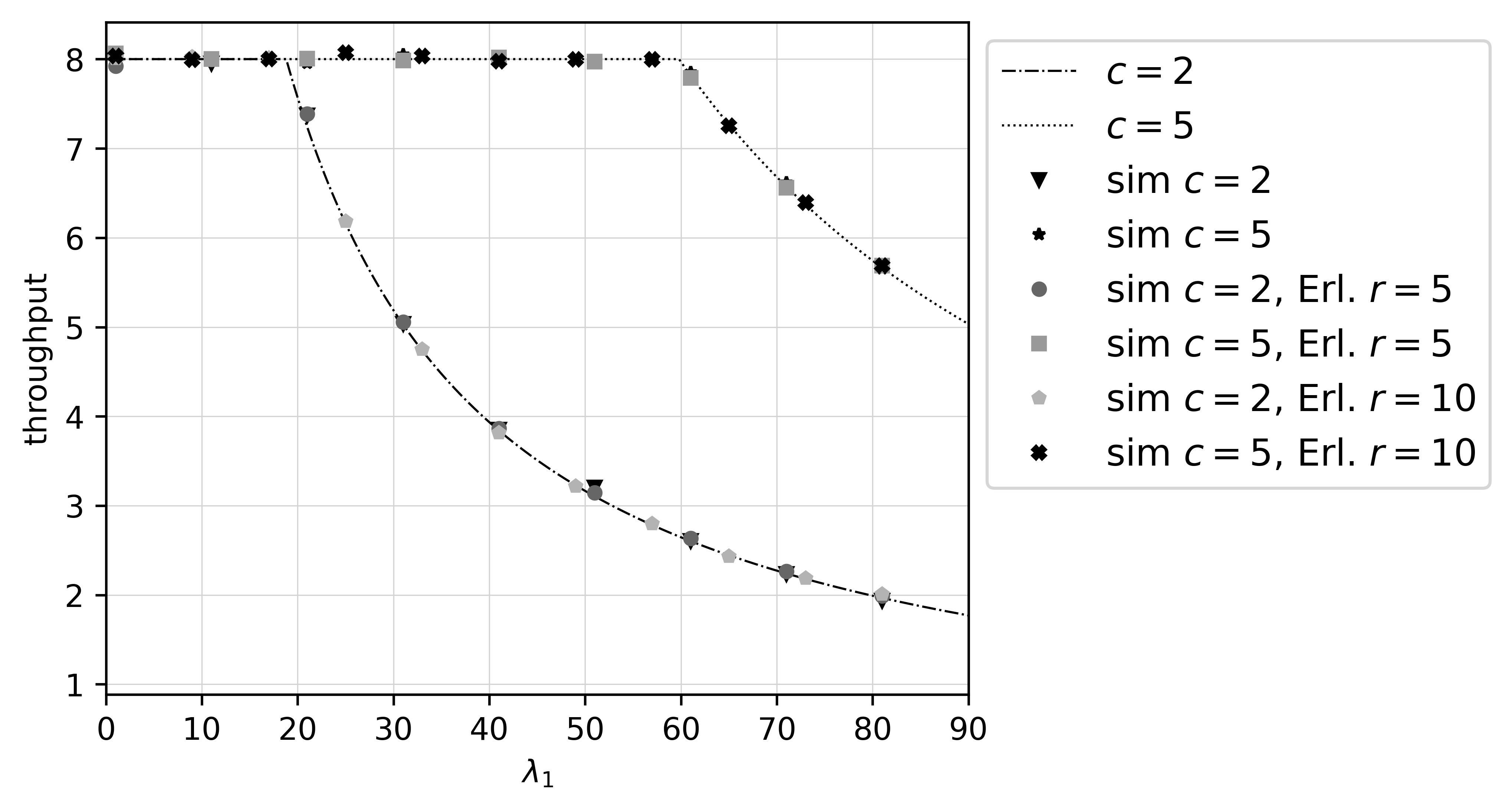} 
\caption{Throughput against arrival rate of class-1 customers ($\lambda_{2} = 8$).}
\label{fig:6}
\end{figure}

\begin{figure}[ht]
\centering
\includegraphics[width=0.8\linewidth]{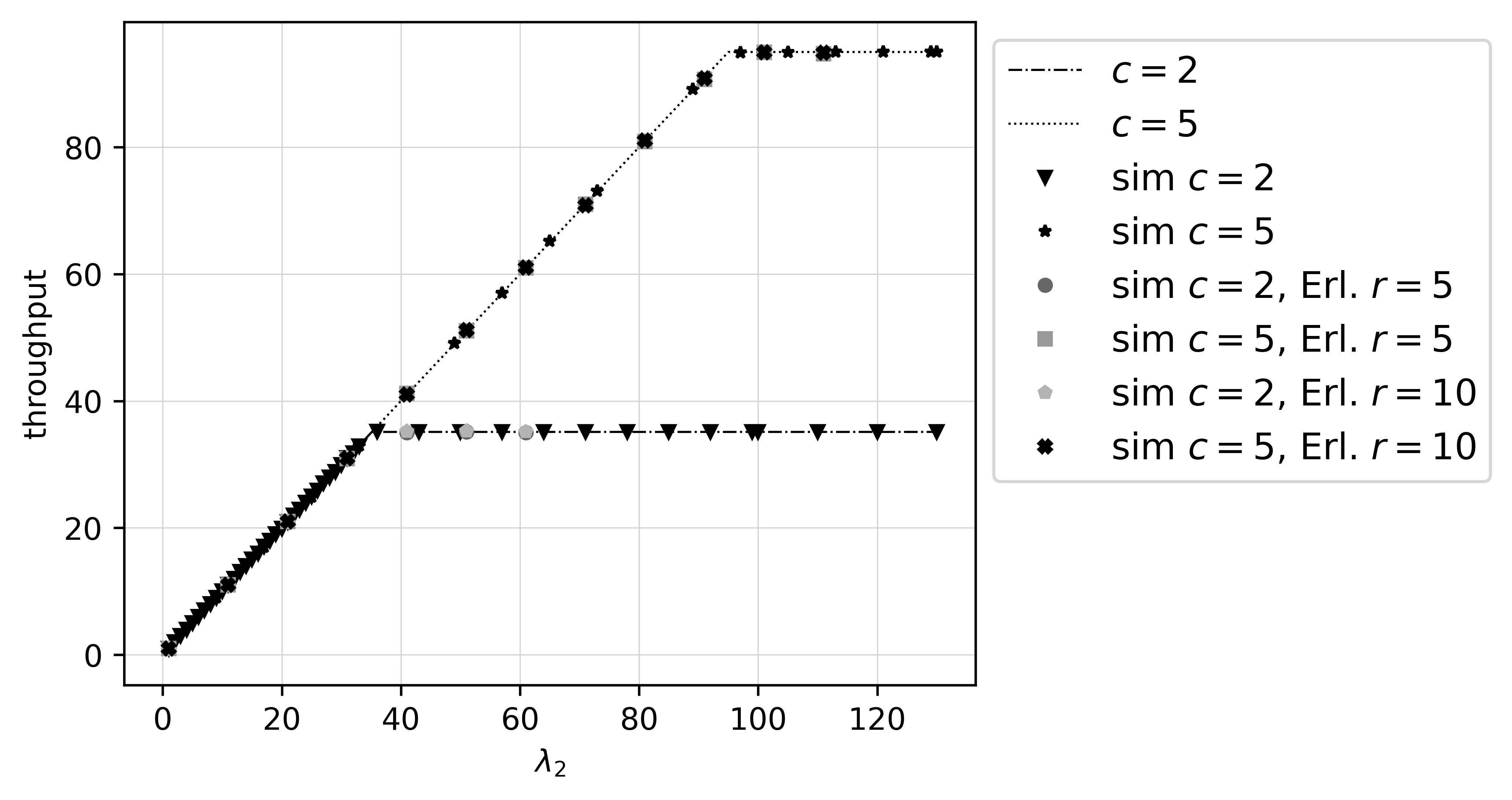} 
\caption{Throughput against arrival rate of class-2 customers ($\lambda_{1} = 1$).}
\label{fig:7}
\end{figure}

Figure \ref{fig:6} reflects changes in the throughput of class-2 customers against $\lambda_1$ when $\lambda_2$ is fixed. It can be seen that the throughput values remain unchanged at $\lambda_2$ when $\lambda_1$ goes up to certain thresholds, then drop as $\lambda_1$ continues to increase. Meanwhile, Figure \ref{fig:7} indicates that the throughput is equal to $\lambda_2$ up to a certain threshold of $\lambda_2$, then remains unchanged at a value $\lambda_{max}$ when class-2 customers arrive at very high rates. At that, $\lambda_{max}$ is defined as 

$$\lambda_{max}=\frac{\sum_{i=0}^{c-1}(c-i)\frac{\rho_1^i}{i!}}{\sum_{i=0}^{c}\frac{\rho_1^i}{i!}}\mu_2,$$
which is the right hand side of the stability condition (\ref{eq:39}).

Finally, it is noticeable that when we let service times of class-1 customers follow Erlang distributions, the results do not change much as compared to the exponential case. This agrees with our stability condition which depends on only the mean service time of class-1 customers.

\section{Conclusion}\label{conclusion:sec}
In this paper, we have considered a modified Erlang system for cognitive radio networks and related applications. We have established a stability condition which is insensitive to the service time distributions of primary and secondary users. This result has implied that the throughput of secondary users is insensitive to the service time distributions of primary and secondary users, provided that the means are fixed. For the case of exponential service time distributions of primary and secondary customers, we have derived some stationary performance measures. Our extensive simulations have shown that the mean waiting time of secondary users increases with the increase in the variance of service time of primary users while the mean number of terminations of secondary users is almost insensitive to the service time distributions of primary users, provided that the means are fixed. Our findings could be used in resource allocation in cognitive radio networks and related applications.

\AcknowledgementSection
The research of EM and SR is  supported by the Russian Foundation for Basic Research, projects  18-07-00147, 18-07-00156 and 19-07-00303. The research of TP and HQN was supported in part by JSPS KAKENHI Grant Number 18K18006.






\begin{references}




\bibitem{3} Akyildiz, I. F.; Lee, W.-Y.; Vuran, M. C.; Mohanty, S. NeXt generation dynamic spectrum access cognitive radio wireless networks: a survey. {\it Comput. Netw.} 2006, 50(13), 2127--2159. 

 \bibitem{Akyildiz08}
Akyildiz, I. F.; Lee, W. Y.; Vuran, M. C.; Mohanty; S.  A survey on spectrum management in cognitive radio networks. {\it IEEE Commun. Mag.} 2008, 46(4), 40-48.

 \bibitem{Letaief09}
Letaief, K. B.; Zhang, W. Cooperative communications for cognitive radio networks. {\it Proc. IEEE.} 2009, 97(5), 878-893.

\bibitem{Mitola}
Mitola, J.; Maguire, G. Q. Cognitive radio: making software radios more personal. {\it IEEE Pers. Commun.} 1999, 6(4), 13-18.


\bibitem{Wang10}
Wang, B.; Liu, K. R. Advances in cognitive radio networks: A survey. {\it IEEE J. Sel. Top. Signal Process.} 2010, 5(1), 5-23.

\bibitem{Ostovar20}
Ostovar, A.; Keshavarz, H.; Quan, Z. Cognitive radio networks for green wireless communications: an overview. {\it Telecommun. Syst.} 2020, 1-10.

\bibitem{Mitrany68}
 Mitrany, I. L.; Avi-Itzhak, B.  A many-server queue with service interruptions. {\it Oper. Res.} 1968, 16(3), 628-638.

\bibitem{Akutsu19}
 Akutsu, K.; Phung-Duc, T. Analysis of retrial queues for cognitive wireless networks with sensing time of secondary users, {\it Lect. Notes Comput. Sci.} 2019, LNCS 11688, 77-91. 
 
 \bibitem{Salameh20}
Salameh, O.; Bruneel, H.; Wittevrongel, S. Performance evaluation of cognitive radio networks with imperfect spectrum sensing and bursty primary user traffic. {\it  Math. Probl. Eng.} 2020.
 
 \bibitem{Salame17}
Salameh, O.; De Turck, K.; Bruneel, H.; Blondia; C.; Wittevrongel, S. Analysis of secondary user performance in cognitive radio networks with reactive spectrum handoff. {\it Telecommun. Syst.}  2017, 65(3), 539-550.

\bibitem{Dudin16} Dudin, A. N.; Lee, M. H.; Dudina, O.; Lee, S. K. Analysis of priority retrial queue with many types of customers and servers reservation as a model of cognitive radio system. {\it IEEE Trans. Commun.} 2016, 65, 186-199. 

\bibitem{konishi}  Konishi, Y.; Masuyama, H.; Kasahara, S.; Takahashi, Y. Performance analysis of dynamic spectrum handoff scheme with variable bandwidth demand of secondary users for cognitive radio networks.  {\it Wirel. Netw.} 19. 2013, 607--617.

\bibitem{osama}  Salameh, O.; De Turck, K.; Bruneel, H.; Blondia, C.; Wittevrongel, S. Analysis of secondary user performance in cognitive radio networks with reactive spectrum handoff.  {\it Telecommun. Syst.} 2017, 65, 539--550.

\bibitem{Paul18}
Paul, S.; Phung-Duc, T. Retrial queueing model with two-way communication, unreliable server and resume of interrupted call for cognitive radio networks. In {\it Information technologies and mathematical modelling. Queueing theory and applications}, Springer, Cham. 2018, pp. 213-224.

\bibitem{Shajin19}
Shajin, D.; Dudin, A. N.; Dudina, O.; Krishnamoorthy, A. A two-priority single server retrial queue with additional items. {\it J. Ind. Manag. Optim.} 2019.

\bibitem{wong} Wong, E. W. A.; Foh, C. H. Analysis of cognitive radio spectrum access with finite user population. {\it IEEE Wireless Commun. Lett.} 2009, 13(5), 294--296.

\bibitem{Zhang19}
Zhang, Y.; Wang, J. ; Li, W. W. Optimal pricing strategies in cognitive radio networks with heterogeneous secondary users and retrials. {\it IEEE Access.} 2019, 7, 30937-30950.

\bibitem{Liu19}
Liu, J.; Jin, S.; Yue, W.  Performance evaluation and system optimization of Green cognitive radio networks with a multiple-sleep mode. {\it Ann. Oper. Res.} 2019, 277, 371-391.


 \bibitem{Azarfar14}
Azarfar, A.; Frigon, J. F.; Sanso, B. Priority queueing models for cognitive radio networks with traffic differentiation. {\it EURASIP J Wirel. Commun. Netw.} 2014, 206.

\bibitem{Dimitriou} Dimitriou, I.;  Phung-Duc, T.  Analysis of cognitive radio networks with cooperative communication. {\it Proceedings of the 13th EAI International Conference on Performance Evaluation Methodologies and Tools}, 2020; pp. 192-195.

\bibitem{Dragieva}
Dragieva, V. I.; Phung-Duc, T. Queueing analysis of cognitive radio networks with finite number of secondary users.  {\it Proceedings of the 25th International Conference on Analytical and Stochastic Modelling Techniques and Applications}, 2019.

\bibitem{Nazarov18}
Nazarov, A.; Phung-Duc, T.; Paul, S. Unreliable single-server queue with two-Way communication and retrials of blocked and interrupted calls for cognitive radio networks. {\it Distributed Computer and Communication Networks}, pp. 276-287 (2018), Springer, Cham.

\bibitem{Kim12}
Kim, K. T-preemptive priority queue and its application to the analysis of an opportunistic spectrum access in cognitive radio networks. {\it Comput. Oper. Res.} 2012, 39(7), 1394-1401.


\bibitem{Vishnevskiy20}
Vishnevskiy, V. M.; Samouylov, K. E.; Yarkina, N. V. Mathematical model of LTE cells with machine-type and broadband communications. {\it Autom. Remote Control.}  2020, 81, 622-636.

\bibitem{Smith}
 Smith, W.L. Regenerative stochastic processes. Proc. R. Soc. Lond. A (232), 6-31 (1955).
 

\bibitem{Asmus}
 Asmussen, S. Applied probability and Queues. 2nd edn. Springer, Springer-Verlag New York (2003).
 

 
 
\bibitem{Morozov97}
{Morozov, E.} The tightness in the ergodic analysis of regenerative
queueing processes. {\it Queueing Syst.} 1997, 27, 179-203.

\bibitem{Morozov04}
{Morozov, E.} Weak regeneration in modeling of queueing processes.
{\it Queueing Syst.} 2004,  46, 295-315.

\bibitem {Rosario}
 Morozov E., Delgado R.  Stability analysis of regenerative queues, {\it Autom. Remote
Control.} 2009, 70, 1977-1991.

\bibitem{Feller}
  Feller, W.  An introduction to probability theory  vol.2,  Wiley, New York (1971).
  
  

 \bibitem{Serfozo}
 Serfozo, R. Basics of applied stochastic processes, Springer-Verlag (2009).

 \bibitem{Neuts}
 Neuts, M.F. Matrix geometric solutions in stochastic models—an algorithmic approach. Johns Hopkins University Press, Baltimore, MD (1981).




\bibitem{Phung-Duc10}
Phung-Duc, T.; Masuyama, H.; Kasahara, S.; Takahashi, Y. A simple algorithm for the rate matrices of level-dependent QBD processes. {\it Proceedings of the 5th International Conference on Queueing Theory and Network Applications}, 2010; pp. 46-52.




 


 







































 
\end{references}
\end{document}